\documentclass[12pt]{article}
\usepackage{subfigure,a4}
\usepackage{multirow}
\usepackage{epsfig,ulem}
\usepackage{dcolumn}
\usepackage{bm}

\newcommand{\be}{\begin{eqnarray}}
\newcommand{\ee}{\end{eqnarray}}

\usepackage[english]{babel}
\usepackage{amssymb}
\usepackage{amsmath}
\usepackage{amsfonts}
\usepackage{amsbsy}
\usepackage{indentfirst}
\usepackage{graphicx}
\usepackage{color}

\newcommand{\R}{\mathbb{R}}

\newcommand{\e}{\mathrm{e}}

\newtheorem{theorem}{Theorem}[section]

\newenvironment{proof}[1][Proof]{\noindent\textbf{#1.} }{\ \rule{0.5em}{0.5em}}
\def\d{{\rm d}}

\def\diag{{\rm diag}}
\def\<{\langle}
\def\>{\rangle}

\def\diag{{\rm diag}}

\def\ga{\gamma}

\newcommand{\beq}{\begin{equation}}
\newcommand{\eeq}{\end{equation}}

\newcommand{\bmat}{\begin{displaymath}}
\newcommand{\emat}{\end{displaymath}}

\def\1{{\bf 1}}
\begin{document}

\title{Spectral analysis of a selected non self-adjoint Hamiltonian in an infinite dimensional Hilbert space}
\author{
N. Bebiano\footnote{ CMUC, University of Coimbra, Department of
Mathematics, P 3001-454 Coimbra, Portugal (bebiano@mat.uc.pt)},
J.~da Provid\^encia\footnote{CFisUC, University of Coimbra,
Department of Physics, P 3004-516 Coimbra, Portugal
(providencia@teor.fis.uc.pt)}~ and J.P. da
Provid\^encia\footnote{Depatamento de F\'\i sica, Univ. of Beira
Interior, P-6201-001 Covilh\~a, Portugal
(joaodaprovidencia@daad-alumni.de)}}

\maketitle
\begin{abstract}
The so-called equation of motion method
is useful to obtain the explicit form of the eigenvectors and
eigenvalues of certain non self-adjoint bosonic Hamiltonians with
real eigenvalues.
These operators can be diagonalized when they are expressed in terms
of pseudo-bosons, which do not behave as ordinary bosons under the
adjoint transformation, but obey the Weil-Heisenberg commutation
relations.
\end{abstract}
\section{Introduction and preliminaries}
 In conventional formulations of
non-relativistic quantum mechanics,
the Hamiltonian operator is self-adjoint. However, certain
relativistic extensions of quantum mechanics lead to the
consideration of non self-adjoint Hamiltonian operators
with a real spectrum. This motivated an
intense research activity, both on the physical and mathematical
level (see, e.g.
\cite{bagarello,providencia,scholtz,[1],[2],[3],mostafa} and their
references).

 Throughout, we shall use synonymously the terms Hermitian and
self-adjoint. Denote  by $L^2(\R^2)$ the Hilbert space of square
integrable functions in two real variables, endowed with the
standard inner product $$\langle
f,g\rangle=\int_{-\infty}^{+\infty}\int_{-\infty}^{+\infty}
g(x,y)\overline{ f(x,y)} \d x\d y.$$ For
\begin{equation}\label{calD}{\cal D}=\left\{f(x,y)\e^{-(x^2+y^2)}:f(x,y)~\text{is a
polynomial
in}~x,y\right\}
\end{equation}which is a dense domain in $L^2(\R^2)$ \cite{bebiano},
let $a,b:{\cal D}\rightarrow{\cal D}$ be bosonic operators defined,
respectively, by
$$a=x+{1\over2}{\partial\over\partial x},\quad
b=y+{1\over2}{\partial\over\partial y}.$$ Consider also their
adjoints $a^*,b^*:{\cal D}\rightarrow{\cal D}$
$$a^*=x-{1\over2}{\partial\over\partial x}, \quad b^*=y-{1\over2}{\partial\over\partial y}.$$
We recall that, conventionally, $a,b$ are said to be {\it
annihilation operators}, while $a^*,b^*$ are  {\it creation
operators}. It is worth noticing that these operators are unbounded.
Moreover, $\cal D$ is stable under the action of $a,b$ and of their
adjoints, and they satisfy the commutation rules (CR's),
\begin{equation}\label{CR1}[a,a^*]=[b,b^*]=\1,\end{equation}where $\1$ is the identity operator in
$\cal D.$ (This means that $aa^*f-a^*af=bb^*f-b^*bf=f$ for any
$f\in{\cal D}.$) Furthermore,
\begin{equation}\label{CR2}[a,b^*]=[b,a^*]=[a^*,b^*]=[a,b]=0.\end{equation}
As it is well-known, the canonical commutation relations (\ref{CR1})
and (\ref{CR2}) characterize an algebra of Weil-Heisenberg (W-H).
Moreover, the following holds,
$$a\Phi_0=b\Phi_0=0,$$
for $\Phi_0=\e^{-(x^2+y^2)}\in {\cal D}$, a so-called {\it vacuum
state}.
The set of vectors
\begin{equation}\label{Phimn}\{\Phi_{m,n}=a^{*m}b^{*n}\Phi_0:~m,n\geq0\},\end{equation}constitutes a
basis of $\cal H$, that is, every vector in $L^2(\R^2)$ can be
uniquely expressed in terms of this vector system, which is {\it
complete}, since 0 is the only vector orthogonal to all its
elements.

The main goal of this note is to investigate spectral properties of
certain non self-adjoint operators which are expressed as quadratic
combinations of bosonic operators.
\section{Non self-adjoint Hamiltonian
constructed in terms of $su(1,1)$ generators}

A linear operator acting on a finite dimensional Hilbert space which
is non self-adjoint and has distinct real eigenvalues, is similar to
its adjoint operator. Concretely, the following holds, having in
mind that the spectrum of a finite matrix is the set of its
eigenvalues. \color{black}
\begin{theorem}\label{theorem} Let $H$ be
an  $n\times n$ complex non self-adjoint matrix with  distinct real
eigenvalues. Then $H$ and $H^*$ have a common spectrum, a complete
systems of eigenvectors and they are similar, that is, there exists
a unitary matrix $S$ such that $H^* =SHS^{-1}.$ Moreover, if
$H\phi_i=\lambda_i\phi$, $H^*\psi_i=\lambda_i\psi_i$,
$i=1,\ldots,n,$ then the eigenvectors may be normalized so that
$\langle\psi_i,\phi_j\rangle=\delta_{ij}
.$\end{theorem}
\begin{proof}Since the eigenvalues of $H$ are real, they coincide
with those of $H^*$. Let
$\sigma(H)=\sigma(H^*)=\{\lambda_1,\dots,\lambda_n\}$, and let
$\phi_i,~\psi_i,~i=1,\ldots,n$ be such that
$$H\phi_i=\lambda_i\phi_i,\quad H^*\psi_i=\lambda_i\psi_i.$$
Further, as the eigenvalues are mutually distinct, the corresponding
eigenvectors are linearly independent, and so the eigensystems
$\{\psi_i\},~\{\phi_j\}$ corresponding to the eigenvalues
$\lambda_i,~i=1,2,\ldots,n,$ are complete.


Let $S$ be defined by $\psi_i=S\phi_i.$ Thus, $S$ is unitary and
$$H^*S\phi_i=\lambda_iS\phi_i=SH\phi_i,~i=1,\ldots,n,$$
implying that $H$ and $H^*$ are unitarily similar.

Having in mind that $$\langle H\phi_i,\psi_j\rangle=\langle
\phi_i,H^*\psi_j\rangle=\lambda_i\langle
\phi_i,\psi_j\rangle=\lambda_j\langle \phi_i,\psi_j\rangle,$$the
orthonormality relation follows.
\end{proof}\\

The following question naturally arises.
 Do the above properties survive in the infinite dimensional
setting? We answer this question for the model described by the
linear operator $H:L^2(\R^2)\rightarrow L^2(\R^2)$ defined as
\begin{equation}\label{model}H=a^*a+bb^*+\beta(a^*a-bb^*)
+\gamma(a^*b^*-ab),\quad\beta,\gamma\in\R,\end{equation} which
treats a system of two interacting bosons. It is obvious that $H$ is
non self-adjoint. \color{black}

Consider the class $\cal C$ of unbound linear operators $H$ acting
on the infinite dimensional Hilbert space $\cal H$ and satisfying
the following properties:

\noindent (I) $H$ is non self-adjoint,

\noindent (II) $H$ has real eigenvalues,


\noindent (III) $H$ and $H^*$ are isospectral.

\noindent (IV) The associated eigenvectors form biorthogonal bases
of $\cal H$.

\noindent (V) There exists a unitary transformation $S$ on $\cal H$
such that $H^*=SHS^{-1}$.\\

Throughout we prove that $H$ in (\ref{model}) belongs to the class
$\cal C$, and so, in the eigenvectors bases, $H$ and $H^*$ have  a
diagonal representation.
\section{The equation of motion method for pseudo-bosonic operators}
 Let us consider
a non self-adjoint
Hamiltonian $H$ ($H\neq H^*$) which
is expressed
in terms of unbounded bosonic operators $a_1,\ldots,a_n,$
$a^*_1,\ldots,a^*_n,$ that is, linear operators acting on $\cal D$
and satisfying the CR's,
\begin{equation}\label{WHCR}[a_i,a^*_j]=\delta_{ij},~[a_i,a_j]=0,~[a^*_i,a^*_j]=0.\end{equation}
 As usual, $\delta_{ij}$ denotes  the Kroenecker symbol (which
equals 0 for $i\neq j$ and $1$ for $i=j$).

Let $\cal V$ be the linear space spanned by
$a_1,\ldots,a_n,a^*_1,\ldots,a^*_n$ and let $f\in{\cal V}$, that is,
$f$ is a linear combination of the operators
$a_1,\ldots,a_n,a^*_1,\ldots,a^*_n$. Having in mind that
$$[a_ia_j,a^*_k]=\delta_{ik}a_j+\delta_{jk}a_i,~
[a^*_ia_j,a_k]=\delta_{ik}a_j$$ and their adjoint relations, it is
clear that $[H,f]\in{\cal V}$. Obviously, the operator $[H,f]: {\cal
V}\rightarrow{\cal V}$ is a linear operator on $f$ and next we prove
that its eigenvalues are closely related with those of $H$.

\begin{theorem}
Let $H$ be a non self-adjoint operator acting on $\cal H$ with real
eigenvalues expressed in terms of boson operators satisfying
(\ref{WHCR}). Let $f\in{\cal V}$ and $\lambda\in\R$ be such that
\begin{equation}\label{EMM}[H,f]=\lambda f.\end{equation}
Suppose $\Lambda\in\R$ and $\psi\in{\cal H}$ are, respectively,  an
eigenvalue of $H$ and an associated eigenvector. If $f\psi\neq0$,
then also $\Lambda+\lambda$ is an eigenvalue of $H$ and $f\psi$ a
corresponding eigenvector. Moreover, 
there exists $0\neq g\in{\cal V}$ such that
\begin{equation*}[H,g]=-\lambda g.\end{equation*}
\end{theorem}
\begin{proof}Under the hypothesis, we have $H\psi=\Lambda\psi$.
Then,
$$Hf\psi=[H,f]\psi+fH\psi=(\Lambda+\lambda)f\psi.$$ Since $\Lambda$ and
$\lambda+\Lambda$ are both real, so is $\lambda$.
Since $[H^*,f^*]=-\lambda f^*$ and $H$ and $H^*$ have the same
eigenvalues, it follows that there exists $0\neq g\in{\cal V}$  such
that
$[H,g]=-\lambda g$. Observe that the operator $[H,f]: {\cal
V}\rightarrow{\cal V}$ is represented by some $2n\times 2n$ matrix
$T$, while the operator $[H^*,f]: {\cal V}\rightarrow{\cal V}$ is
represented by the transconjugate matrix $T^*$, and so the
eigenvalues of $T$ and $T^*$ coincide since they are real.
\end{proof}\\

If $\lambda>0$, $f$ is called an {\it excitation operator} ({\it
deexcitation operator} if $\lambda<0$).

The so called equation of motion method (EMM) consists in the
solution of eq. (\ref{EMM}), with the determination of the
eigenvalues of $H$ and corresponding eigenvectors.

Let $0\neq g\in{\cal V}$ be such that $[H,g]=\lambda' g$. The Jacobi
identity yields
$$[[H,f],g]+[[f,g],H]+[[g,H],f]=0.$$
Now,  $[[f,g],H]=0$, since $[f,g]$ is a multiple of the identity.
Thus,
$$(\lambda+\lambda')[f,g]=0.$$ As a consequence, either $\lambda=-\lambda'$ or
$[f,g]=0.$ If $[f,g]\neq0$, which happens only if
$\lambda=-\lambda'$,  $f$ and $g$ may be normalized so that
$[f,g]=1.$ Since $f^*\neq g$, the operators obtained by this
procedure do not describe ``ordinary" bosons. In this case, it  may
be seen that there exists $0\neq\Psi_0\in{\cal D}$ such that either
$f\Psi_0=0$ or $g\Psi_0=0$. Assume that $f\Psi_0=0$. We write
$g=f^\ddag$ and it can be checked that the vectors system
$\{f^{\ddag~ n}\Psi_0:n\geq0\}$ constitutes a complete basis of
$L^2(\R^2)$. There also exists $\Psi_0'\in{\cal D}$ such that
$g^*\Psi_0'=f^{\ddag*}\Psi_0'=0$ and the vector system $\{f^{*~
n}\Psi'_0:n\geq0\}$ is an eigenvector system for $H^*$, which
constitutes a complete basis of $L^2(\R^2)$. The bases $$\{f^{\ddag~
n}\Psi_0:n\geq0\}~~\text{ and}~~ \{f^{*~ n}\Psi'_0:n\geq0\}$$ are
biorthogonal:$$\langle f^{*n}\Psi'_0,f^{\ddag
m}\Psi_0\rangle=\langle\Psi'_0,f^nf^{\ddag
m}\Psi_0\rangle=\delta_{nm}m!\langle\Psi'_0,\Psi_0\rangle.$$ Thus,
the operators $f,f^\ddag$ are said to describe {\it pseudo-bosons}
for $H$, while the operators $f^{\ddag*},f^*$ describe pseudo-bosons
for $H^*$ (cf. Bagarello \cite{bagarello1,bagarello}).
If $g\Psi_0=0,$ with $\Psi_0\in {\cal D}$, an analogous discussion
takes place.
\section{$H$ belongs to the class $\cal C$}\label{S4} Throughout we prove
that $H$ in eq. (\ref{model}) belongs to the class $\cal C$. As
already observed, $H$ satisfies (I).

(II) $H$ has real eigenvalues.

The operator $H$ is quadratic in the bosonic operators and we may
explicitly find its eigenvalues by the equations of motion method,
as follows. We look for a linear combination of the bosonic
operators such that its commutator with $H$ satisfies the
proportionality condition,
$$[H,(x_aa^*+x_bb^*+y_aa+y_bb)]=\lambda(x_aa^*+x_bb^*+y_aa+y_bb),~\lambda\in\R.$$
This leads to the following system of linear equations in
$x_a,x_b,y_a,y_b$ and $\lambda$:
\begin{eqnarray*}(1+\beta)x_a-\gamma y_b=\lambda x_a,\\(1-\beta)x_b-\gamma y_a=\lambda
x_b,\\-(1+\beta)y_a-\gamma x_b=\lambda y_a,\\-(1-\beta)y_b-\gamma
x_a=\lambda y_b.
\end{eqnarray*}
Solving this  linear system, we readily obtain
$\lambda=\pm\beta\pm\sqrt{1+\gamma^2}.$ The vectors $(x_a,x_b,$
$y_a,y_b)^T$ and the parameters $\lambda$ are, respectively, the
eigenvectors and eigenvalues of the $4\times4$ real symmetric matrix
$$T=\left[\begin{matrix}1+\beta&0&0&-\gamma\\0&1-\beta&-\gamma&0\\0&-\gamma&-1-\beta&0\\ -\gamma&0&0&-1+\beta\end{matrix}\right].$$
The eigenvalues of $T$ and associated eigenvectors are given by

\begin{eqnarray*}&&\lambda_1=-\beta-\sqrt{1+\ga^2},\quad v_1={\cal N}({0, -1 +
\sqrt{1 + \ga^2}, \ga,
0})^T,\quad\\&&\lambda_2=\beta-\sqrt{1+\ga^2},\quad v_2={\cal N}({-1
+ \sqrt{1 + \ga^2}, 0, 0,
\ga})^T,\quad\\&&\lambda_3=-\beta+\sqrt{1+\ga^2},\quad v_3={\cal
N}({0, 1 + \sqrt{1 + \ga^2}, -\ga,
0})^T,\quad\\&&\lambda_4=\beta+\sqrt{1+\ga^2},\quad v_4 ={\cal N}(1
+ \sqrt{1 + \ga^2}, 0, 0, -\ga)^T~.\end{eqnarray*}
The determination  of the normalization constant ${\cal N}$ is
postponed and done according to future convenience.

Consider now the operators
\begin{eqnarray*}&&c={\cal N}({ (-1 +
\sqrt{1 + \ga^2})b^*+ \ga a }) ,\\&& d={\cal N}(({-1 + \sqrt{1 +
\ga^2})a^*+ \ga b}) ,\\&& d^\ddag={\cal N}({ (1 + \sqrt{1 +
\ga^2})b^*- \ga a}),
\\&&c^\ddag ={\cal N}((1 +
\sqrt{1 + \ga^2})a^*- \ga b),
\end{eqnarray*}
also defined on $\cal D$.
Notice that
$c^\ddag\neq c^*$ and $d^\ddag\neq d^*$.

Let us take $${\cal
N}=\left(2\gamma\sqrt{1+\gamma^2}\right)^{-1/2}.$$ Then,
the operators $c^\ddag,d^\ddag,c,d$, satisfy the CR's of a
Weil-Heisenberg algebra,
$$[c,c^\ddag]=[d,d^\ddag]=\1,$$vanishing all the remaining commutators between them.
Moreover, $H$ can be written as
$$H=
\beta(c^\ddag c-d^\ddag d)+\sqrt{1+\gamma^2}(c^\ddag c+dd^\ddag).$$
Next, we show that the Hamiltonian presents a diagonal form when it
is expressed in terms of the operators $c^\ddag,d^\ddag,c,d$.

Having in mind that $$\Phi_0=\exp\left(-(x^2+y^2)\right)\in{\cal
D}$$ satisfies $a\Phi_0=b\Phi_0=0$, it may be easily verified that
$$\Psi_0=\exp(-\alpha a^*b^*)\Phi_0=\exp\left(-{1+\alpha^2\over1-\alpha^2}(x^2+y^2)-{4\alpha\over1-\alpha^2}xy\right),$$
where
$\alpha=\gamma/(1+\sqrt{1+\gamma^2})=(-1+\sqrt{1+\gamma^2})/\gamma$,
satisfies $$c\Psi_0=d\Psi_0=0.$$ Of course, $\Psi_0\in {\rm span}
{\cal D}$. Let
$$\{\Psi_{m,n}=c^{\ddag m}d^{\ddag n}\Psi_0:~m,n\geq0\}.$$
From the CR's of a Weil-Heisenberg algebra, it follows that
$${H}\Psi_{m,n}=\left(\sqrt{1+\gamma^2}+ m(\beta+\sqrt{1+\gamma^2})+
n(-\beta+\sqrt{1+\gamma^2})\right)\Psi_{m,n}.$$ Henceforth, the
eigenvalues 
of $H$ are given  by \begin{equation}\label{eigenvalues}E_{m,n}=
\sqrt{1+\gamma^2}+ m(\beta+\sqrt{1+\gamma^2})+
n(-\beta+\sqrt{1+\gamma^2}).\end{equation}Thus, (II) holds.
Moreover, the associated eigenvectors to $E_{mn}$ are
$$\Psi_{m,n}=c^{\ddag m}d^{\ddag n}\Psi_0,~m,n\geq0.$$
These eigenvectors constitute a basis of $\cal H$. Further, it can
be shown that this basis is complete \cite{bebiano}.

In order to prove that $H^*$ and $H$ have the same eigenvalues, let
us consider the state vector
$$\Psi'_0=\exp(\alpha a^*b^*)\Phi_0.$$
Clearly, $$c^{\ddag*}\Psi'_0=d^{\ddag*}\Psi'_0=0.$$ It follows that
the state vectors
$$\Psi'_{m,n}=c^{* m}d^{* n}\Psi'_0,~m,n\geq0,$$
are eigenvectors of $H^*$,
$${H^*}\Psi'_{m,n}=\left(\sqrt{1+\gamma^2}+
m(\beta+\sqrt{1+\gamma^2})+
n(-\beta+\sqrt{1+\gamma^2})\right)\Psi'_{m,n}.$$ Hence, $H^*$ and
$H$ are isospectral and (III) holds.

Next, we show that the basis constituted by the eigenvectors of $H$
is orthogonal to the basis of the eigenvectors of $H^*$. The basis
$$\{\Psi'_{m,n}=c^{*m}d^{*n}\Psi'_0: m,n\geq0\},$$ is orthogonal to
the basis $\{\Psi_{m,n}: m,n\geq0\}$, as we have
\begin{eqnarray*}&&\langle\Psi_{m,n},\Psi'_{p,q}\rangle=\langle
c^{\ddag m}d^{\ddag n}\Psi_0,c^{*p}d^{*q}\Psi'_0\rangle=\langle
c^{p}d^{q} c^{\ddag m}d^{\ddag
n}\Psi_0,\Psi'_0\rangle=m!n!\delta_{mp}\delta_{n,q}\langle\Psi_0,\Psi'_0\rangle.\end{eqnarray*}\color{black}
Henceforth, (IV) is satisfied.

Finally, (V) holds. The unitary transformation
$$S=\exp\left({-i{\pi\over2}(a^*a+b^*b)}\right),$$
acts on $a$ and $b$ as follows:
$$ a\rightarrow ia=SaS^{-1}, \quad
b\rightarrow ib=SbS^{-1},$$ and so $S$ takes $H$ into $H^*$, that
is, $$H^*=SHS^{-1}.$$
Hence, $H$ satisfies (V).\\

Some considerations are in order.

1. In the above considered  eigenvectors system, the matrix of $H$
has the following representation
\begin{eqnarray*}&&\left[\begin{matrix}\rho&0&0&0&\ldots\\0&\beta+2\rho&0&0&\ldots\\0&0&2\beta+3\rho&0&\ldots\\0&0&0&3\beta+4\rho&\ldots
\\\vdots&\vdots&\vdots&\vdots&\ddots\end{matrix}\right]\otimes I\\&&+ I\otimes
\left[\begin{matrix}0&0&0&0&\ldots\\0&-\beta+\rho&0&0&\ldots\\0&0&2(-\beta+\rho)&0&\ldots\\0&0&0&3(-\beta+\rho)&\ldots
\\\vdots&\vdots&\vdots&\vdots&\ddots\end{matrix}\right],\end{eqnarray*}
where $I$ is the semi-infinite identity matrix,
$\rho=\sqrt{1+\gamma^2}$ and $\otimes$ denotes the usual Kroenecker
product. Hence, the matrix of $H$ is the following diagonal matrix
\begin{eqnarray*}
\bigoplus_ {k=0}^\infty
\left[\begin{matrix}k(\beta+\rho)+\rho&0&0&\ldots
\\0&k(\beta+\rho)-\beta+2\rho&0&\ldots
\\0&0&k(\beta+\rho)-2\beta+3\rho&\ldots
\\\vdots&\vdots&\vdots&\ddots\end{matrix}\right].\end{eqnarray*}

2. In the consstructed eigenvectors bases, $H$ and $H^*$ have  a
diagonal representation.

3. The integers $m,n$ are the number of pseudo-bosons of type
$c^\ddag,d^\ddag$, respectively. Since, however, $c^\ddag\neq
c^*,~d^\ddag\neq d^*$, they do not describe ``ordinary" dynamical
bosons and they are said to describe dynamical pseudo-bosons
\cite{bagarello1,bagarello}.

Remark that, although  $H$ is non self-adjoint, it has real
eigenvalues and a complete system of eigenvectors for any
$\gamma,\beta\in\R$. If we replace $\gamma\in \R$ by $i\lambda$, $H$
becomes self-adjoint, but then it only has real eigenvalues and a
complete system of eigenvectors if $|\gamma|<1.$ The existence of
complex eigenvalues for $|\ga|>1$ indicates that the system under
consideration is not stable, or that $H$ is not bounded from below,
violating conditions which, in general, are required on physical
grounds.
\section{Invariant subspaces}\label{S6}
In this Section, we provide an alternative approach to determine the
eigenvalues and eigenvectors of $H$.

\color{black} We notice
that the operator $H$ is expressed as a linear combination of the
generators of the $su(1,1)$ algebra
$$a^*a+bb^*,~ab,~a^*b^*,$$
and of the operator
$$a^*a-b^*b,$$
called, for simplicity, {\it Casimir operator}, since it is closely
related to the actual Casimir operator which is given by
$$C:=(a^*a+bb^*)^2-2a^*b^*ab-2aba^*b^*=(a^*a-b^*b-\1)(a^*a-b^*b+\1).$$
This is the key fact for the alternative procedure for the spectral
analysis in this Section.

Since any $f\in L^2(\R^2)$ may be expanded in the basis
$\{\Phi_{mn}:m,n\geq0\}$, we may identify $L^2(\R^2)$ with ${\cal
H}={\rm span}\{\Phi_{mn}:m,n\geq0\}$. It is convenient to express
${\cal H}$ as a direct sum of eigenspaces of the Casimir operator,
$${\cal H}=\bigoplus_{k=-\infty}^{+\infty}{\cal H}_k,\quad {\cal H}_k={\rm span}\{\Phi_{mn};m-n=k,m,n\geq0\}.$$
Notice that  eigenspaces of the Casimir operator are invariant
subspaces of the generators of the $su(1,1)$ algebra,  and
consequently of $H$. In ${\cal H}_k$, the generator $a^*b^*$ is
matricially represented by
$$A_+=\left[\begin{matrix}0&0&0&0&\ldots\\\sqrt{|k|+1}&0&0&0&\ldots\\0&\sqrt{2(|k|+2)}&0&0&\ldots\\0&0&\sqrt{3(|k|+3)}&0&\ldots
\\\vdots&\vdots&\vdots&\vdots&\ddots\end{matrix}\right],$$
while the generator $ab$ is represented by
$$A_-=A_+^T
,$$
and the generator $a^*a+bb^*$
is represented by
$$A_0=\left[\begin{matrix}|k|+1&0&0&0&\ldots\\0&|k|+3&0&0&\ldots\\0&0&|k|+5&0&\ldots\\0&0&0&|k|+7&\ldots
\\\vdots&\vdots&\vdots&\vdots&\ddots\end{matrix}\right].$$

The following CR's of the $su(1,1)$ algebra are satisfied
\begin{equation}\label{CR3}[A_-,A_+]=A_0,\quad [A_0,A_+]=2A_+,\quad
[A_0,A_-]=-2A_-. \end{equation}The operator $A_+$ ($A_-$) is said to
be a {\it raising (lowering)} operator. That is, if $\Phi$ is an
eigenvector  of $A_0$, so that $A_0\Phi=\Lambda\Phi$, then $A_+\Phi$
is an eigenvector associated with an upward shifted eigenvalue,
$$A_0A_+\Phi=(\Lambda+2)A_+\Phi.$$ Similarly, $A_-\Phi\neq0$ is an
eigenvector associated with a downward shifted eigenvalue,
$$A_0A_-\Phi=(\Lambda-2)A_-\Phi.$$ The spectrum of $A_0$ is bounded
from below and the eigenvector $\Phi_0=(1,0,0,\ldots)^T$ satisfies
$A_-\Phi_0=0$, being called 
a {\it lowest weight state}.  A set of eigenvectors associated with
eigenvalues of $A_0$, which 
are positive, is obtained by acting successively with $A_+$ on
$\Phi_0$. We observe that
$A_0=A_0^*$ and $A_+= A_-^*$. 
It is interesting that, in ${\cal H}_k$, the Hamiltonian $H$ is
represented by a so called {\it pseudo-Jacobi matrix}, that is, a
Jacobi matrix pre or pos multiplied by the involution matrix
$\diag(1,-1,1,-1,\ldots)$,
$$H=\left[\begin{matrix}\beta k+|k|+1&-\ga\sqrt{|k|+1}&0&0&\ldots\\
\ga\sqrt{|k|+1}&\beta k+|k|+3&-\ga\sqrt{2(|k|+2)}&0&\ldots
\\0&\ga\sqrt{2(|k|+2)}&\beta k+|k|+5&-\ga\sqrt{3(|k|+3)}
&\ldots\\0&0&\ga\sqrt{3(|k|+3)}&\beta k+|k|+7&\ldots
\\\vdots&\vdots&\vdots&\vdots&\ddots\end{matrix}\right].$$
{For simplicity, this matrix has been denoted by the same symbol
$H$. It is clear that $H$ is unitarily similar to its transpose
$H^T$,
$$H^T=\e^{i{(\pi/2)}A_0}H\e^{-i{(\pi/2)}A_0}.$$} The diagonalization of this matrix
provides an
alternative procedure to the previous approach.

\color{red}For completeness, it may be in order to observer that it
is also possible to represent the generator $ab$ by
$$A'_+=
\left[\begin{matrix}0&\sqrt{|k|+1}&0&0&\ldots\\0&0&\sqrt{2(|k|+2)}&0&\ldots\\0&0&0&\sqrt{3(|k|+3)}&\ldots\\0&0&0&0&\ldots
\\\vdots&\vdots&\vdots&\vdots&\ddots\end{matrix}\right]
,$$ the generator $a^*b^*$ by
$$A'_-={A'_+}^T
,$$
 and the generator $-a^*a-bb^*$   by
$$A'_0=\left[\begin{matrix}-|k|-1&0&0&0&\ldots\\0&-|k|-3&0&0&\ldots\\0&0&-|k|-5&0&\ldots\\0&0&0&-|k|-7&\ldots
\\\vdots&\vdots&\vdots&\vdots&\ddots\end{matrix}\right].$$The commutation relations $[A'_-,A'_+]=A_0,\quad [A'_0,A'_+]=2A'_+,\quad
[A'_0,A'_-]=-2A'_-,$ which characterize the $su(1,1)$ algebra, are
clearly satisfied. However, in this case, the role of the lowest
weight state is played by a {\it highest weight state}, such that
$A_+\Phi_0=0.$\color{black}

Let ${\cal V}$ be the linear space spanned by $A_0,A_+,A_-$ and let
$$H_0=H-\beta k I=A_0+\ga (A_+-A_-).$$Obviously, $H_0\in \cal V.$
Let $$[H_0,{\cal V}]:=\{H_0x-xH_0:x\in {\cal V}\}.$$ Consider the
linear operator $[H_0,{\cal V}] : {\cal V}\rightarrow{\cal V}$. We
look for $xA_++yA_-+zA_0$ such that
$$[H_0,(xA_++yA_-+zA_0)]=\lambda (xA_++yA_-+zA_0).$$
The eigenvector $xA_++yA_-+zA_0$ is easily determined by the secular
equation
$$\left[\begin{matrix}2&0&-2\ga\\0&-2&-2\ga\\-\ga&-\ga&0\end{matrix}\right]\left[\begin{matrix}x\\y\\z\end{matrix}\right]
=\lambda\left[\begin{matrix}x\\y\\z\end{matrix}\right].$$ The
eigenvalues and respective eigenvectors of this $3\times3$ matrix
are given by
$$2\sqrt{1 + \ga^2},\quad\left({1 + \sqrt{1 + \ga^2}\over\ga}, {\ga\over1 + \sqrt{1 +
\ga^2}} , -1\right)^T,$$
$$-2\sqrt{1 + \ga^2},\quad\left({-1 + \sqrt{1 + \ga^2}\over\ga}, {\ga\over-1 + \sqrt{1 +
\ga^2}} , 1\right)^T,$$ and, of course, $0,\quad (\ga,-\ga,1)^T.$
Let us consider now the
matrices\begin{eqnarray*}&&B_+={\ga\over2\sqrt{1+\ga^2}}\left({1 +
\sqrt{1 + \ga^2}\over\ga}A_++ {\ga\over1 + \sqrt{1 + \ga^2}} A_-
-A_0\right)\\&&B_-={\ga\over2\sqrt{1+\ga^2}}\left({-1 + \sqrt{1 +
\ga^2}\over\ga}A_++{\ga\over-1 + \sqrt{1 + \ga^2}} A_-+
A_0\right),\\&&B_0={1\over\sqrt{1+\ga^2}}\left(A_0+\ga(A_+-A_-)\right).\end{eqnarray*}
The following CR's of the $su(1,1)$ algebra  are satisfied
$$[B_-,B_+]=B_0,\quad [B_0,B_+]=2B_+,\quad [B_0,B_-]=-2B_-,$$
suggesting that
\begin{equation}\label{sigmaH0}\sigma(H_0)=\sqrt{1+\ga^2}\{|k|+1,|k|+3,|k|+5,\ldots\},
\end{equation}in agreement with the previous results in Section \ref{S4}. We remark that
$B_0\neq B_0^*,~B_+\neq B_-^*$.  The operator $B_+$ ($B_-$) is said
to be a raising (lowering) operator. That is, if $\Psi$ is an
eigenvector  of $B_0$, so that $B_0\Psi=\Lambda\Psi$, then $B_+\Psi$
is an eigenvector associated with an upward shifted eigenvalue,
$B_0B_+\Psi=(\Lambda+2)B_+\Psi$. Similarly, $B_-\Psi\neq0$ is an
eigenvector associated with a downward shifted eigenvalue,
$B_0B_-\Psi=(\Lambda-2)B_-\Psi$. Next, we observe that
$$\Psi_0=\left(1,~-\alpha(|k|+1)^{1/2},~\alpha^2\left(\begin{matrix}|k|+2\\2\end{matrix}\right)^{1/2},~
-\alpha^3\left(\begin{matrix}|k|+3\\3\end{matrix}\right)^{1/2},~\ldots\right)^T$$
satisfies $B_-\Psi_0=0,$ as may be easily checked.
It follows that  the spectrum of $B_0$ is bounded from below. A
laborious but straightforward computation shows that the Casimir
operator reduces to
$$C:=B_0^2-2(B_+B_-+B_-B_+)=B_0^2-2 B_0-4B_+B_-=(k^2-1)I.$$Thus,
$$C\Psi_0=(B_0^2-2 B_0)\Psi_0=(k^2-1)\Psi_0,$$ implying that
$$B_0\Psi_0=(|k|+1)\Psi_0,$$ so  that
$$\sigma(B_0)=\{|k|+1,|k|+3,|k|+5,\ldots\}.$$
This confirms (\ref{sigmaH0}).
\section{Discussion} In Section 2, the diagonalization of
non-Hermitian Hamiltonians with real eigenvalues, which are
expressed as quadratic combinations of bosonic operators, is briefly
analyzed. It is shown that, quite generally, Hamiltonians of this
class are diagonalizable in terms of dynamical pseudo-bosons, which
are determined by the EMM. In Section 3, a non self-adjoint
Hamiltonian which is expressed as a linear combination of $su(1,1)$
generators, is investigated. Its eigenvalues and eigenvectors have
been determined with the help of a matrix $T$ of size 4 that is real
Hermitian and is determined by the EMM. The investigated Hamiltonian
has a complete system of eigenvectors expressed in terms of the
creation and annihilation operators of pseudo-bosons, and is
orthogonal to the system of eigenvectors of the adjoint Hamiltonian.
Complete vector systems constructed in terms of boson creation
operators acting on the associated vacuum state are obtained.
Infinite matrix representations of the Hamiltonian in such systems
are presented.


It would be interesting to consider more general Hamiltonians of the
investigated type. A challenging problem would be to analyze the
existence of infinite dimensional versions in the spirit of Theorem
\ref{theorem}

\end{document}